\documentclass[opre,nonblindrev]{informs3} 

\DeclareMathVersion{bold}

\usepackage{algorithm}  
\usepackage{algpseudocode}
\usepackage{caption}
\usepackage{color}
\usepackage{subfigure}
\usepackage{appendix}
\usepackage{amsmath}
\usepackage{bm}
\usepackage{ulem}

\DoubleSpacedXI

\usepackage{natbib}
 \bibpunct[, ]{(}{)}{,}{a}{}{,}%
 \def\bibfont{\small}%

\usepackage[bookmarks=false]{hyperref}
\hypersetup{
   colorlinks=true,
   linkcolor=blue, 
   citecolor= blue, 
   filecolor= blue, 
   urlcolor=blue, 
	pdfpagemode=FullScreen,
   }

\newcommand{\E}{{\rm E}}

\TheoremsNumberedThrough     
\ECRepeatTheorems
\EquationsNumberedBySection 

\begin{document}



\RUNTITLE{New Additive OCBA Procedures for \\ Robust Ranking and Selection}

\TITLE{New Additive OCBA Procedures for \\ Robust Ranking and Selection}

\ARTICLEAUTHORS{%
    \AUTHOR{Yuchen Wan}
    \AFF{School of Data Science, Fudan University, Shanghai, China\\\EMAIL{ycwan22@m.fudan.edu.cn}}
    \AUTHOR{Zaile Li\footnote{Corresponding author}}
    \AFF{Technology and Operations Management Area, INSEAD, Fontainebleau, France \\\EMAIL{zaile.li@insead.edu}}
    \AUTHOR{L. Jeff Hong}
    \AFF{Department of Industrial and Systems Engineering, University of Minnesota, Minneapolis, Minnesota\\\EMAIL{lhong@umn.edu}}
} 

\ABSTRACT{%
Robust ranking and selection (R\&S) is an important and challenging variation of conventional R\&S that seeks to select the best alternative among a finite set of alternatives. It captures the common input uncertainty in the simulation model by using an ambiguity set to include multiple possible input distributions and shifts to select the best alternative with the smallest worst-case mean performance over the ambiguity set. In this paper, we aim at developing new fixed-budget robust R\&S procedures to minimize the probability of incorrect selection (PICS) under a limited sampling budget. Inspired by an additive upper bound of the PICS, we derive a new asymptotically optimal solution to the budget allocation problem. Accordingly, we design a new sequential optimal computing budget allocation (OCBA) procedure to solve robust R\&S problems efficiently. We then conduct a comprehensive numerical study to verify the superiority of our robust OCBA procedure over existing ones. The numerical study also provides insights on the budget allocation behaviors that lead to enhanced efficiency. 
}%


\KEYWORDS{robust ranking and selection, input uncertainty, OCBA, additive}


\maketitle
\section{Introduction}
Ranking and selection (R\&S) aims to select the best alternative with the smallest (or largest) mean performance among a finite set of simulated alternatives. Due to its broad relevance for supporting decision making in stochastic environments, it has become one fundamental topic in the simulation area. Numerous R\&S procedures (a.k.a. algorithms) have been developed to guide the simulation and selection decisions in a R\&S study. A recent, comprehensive review of R\&S procedures can be found in \cite{hong2021review}. In the past few years, the R\&S literature has witnessed an increasing trend of  using R\&S procedures in artificial intelligence algorithms to help improve the performance; see, e.g., \cite{liu2023efficient} and \cite{zhang2024sample}.

As a simulation-based decision supporting tool, the practical value of R\&S hinges on the proper specification of the input distribution in the underlying simulation model used to simulate the alternatives. Conventionally, R\&S studies implicitly assume that the input distribution is known and accurate. However, in practice, this is not always possible to achieve. Commonly, the input distribution is estimated from limited real-world input data, leading to the curse of input uncertainty. Such phenomenon is well-known in the simulation literature and attracts extensive research attention in recent years \citep{song2015input, song2017input, song2019input, zhou2017simulation}.

To the best of our knowledge, \cite{fan2013robust} make the first attempt to handle input uncertainty for R\&S problems. They borrow ideas from robust optimization \citep{ben2002robust} and formulate the robust R\&S framework. Robust R\&S characterizes the input uncertainty by a discrete ambiguity set which contains several possible input distributions that all fit well the input data. Given the ambiguity set, robust R\&S  seeks to select the best alternative with the smallest worst-case mean performance over the ambiguity set. Likewise to robust optimization, robust R\&S involves a two-layer minimax structure. The inner layer defines each alternative's worst-case distribution (or scenario) leading to its worse-case (largest) mean performance over the input distributions in the ambiguity set and the outer layer incorporates comparisons among the worst-case mean performances of alternatives.  The robust perspective ensures the terminal selection of the best alternative to be moderately appealing in the presence of input uncertainty. 

The two-layer structure makes robust R\&S more difficult to solve than conventional R\&S and necessitates the development of new procedures. \cite{fan2013robust, fan2020distributionally} adopt a fixed-precision formulation of robust R\&S and design procedures to achieve a pre-specified probability of correct selection (PCS). To accomplish this, they extend the traditional pairwise comparison approach \citep{kim2006selecting}. To control the overall probability of incorrect selection (PICS)  of a procedure, they develop two generic upper bounds for the PICS --- one is \textit{multiplicative} and the other is \textit{additive}. The multiplicative bound is over-loose because it accounts for $km-1$ pairwise comparisons among all $km$ scenarios (Here $k$ and $m$ denote the total number of alternatives and the total number of input distributions in the ambiguity set, respectively). In comparison, the additive bound identifies only $k+m-2$ critical pairwise comparisons in the robust R\&S problem, and thus may be much tighter than the multiplicative bound. Based on this intuition, they develop additive two-stage procedures and show that the additive bound requires a smaller budget than the multiplicative version while ensuring the same PCS. However, despite the additive bound's advantage, they fail to develop efficient sequential procedures with the additive bound due to the pairwise elimination structure chosen for the procedure design. Readers may refer to Section \ref{subsec: decomposition} for a detailed introduction of the additive bound.

Alongside the fixed-precision robust R\&S procedures of \cite{fan2013robust, fan2020distributionally}, developing fixed-budget robust R\&S procedures that minimizes the PICS under a limited sampling budget attracts certain interest. Notably, \cite{gao2016optimal, gao2017robust} successfully extend the traditional fixed-budget  optimal computing budget allocation (OCBA) approach to address robust R\&S problems. The OCBA approach, which originates from \cite{chen2000simulation}, is one of the more prevalent approaches to solve fixed-budget R\&S problems \citep{chen2011stochastic}. With a limited simulation budget, the OCBA approach seeks to derive the optimal budget allocation scheme among the alternatives to minimize the resulted PICS. Due to the intractability of the PICS, the approach typically uses simple and decomposed approximations of the PICS as surrogates for the objective, which enables to derive clear closed-form solutions for the budget allocation problem in terms of the true parameter values of the alternatives. Since the true parameter values are unknown, sequential OCBA procedures are implemented to approximate the optimal solutions dynamically. This approach takes an intuitive optimization viewpoint of R\&S and demonstrates a strong performance in solving various problems \citep{xu2015simulation, almomani2016ordinal}. 
Following this venue, \cite{gao2017robust} derive an upper bound for the PICS of robust R\&S problems and solve the optimal allocation solution, based on which they propose the robust OCBA (R-OCBA) procedure. The procedure demonstrates superior performances over several benchmark procedures in numerical experiments.
Besides, there have been several other attempts to extend traditional fixed-budget R\&S approach to solve robust R\&S problems; see, e.g.,  \cite{zhang2016sequential} and \cite{wan2023upper}. 


Interestingly, \cite{gao2017robust} also notice the advantage of being ``additive'' as \cite{fan2020distributionally} do. From the derived solution to the optimal budget allocation problem, \cite{gao2017robust} observe that the sampling budget allocation should concentrate only on {$k+m-1$} scenarios including the $m$ scenarios of the best alternative and the worse-case scenario of each non-best alternative. 
They then interpret the performance of their R-OCBA procedure in terms of this intuition. However, despite this understanding, they do not provide numerical evidence to support this. Moreover, it can be observed that the PICS upper bound on which they base the analysis and the procedure development is multiplicative. This may explain why their R-OCBA procedure is much more complex than the original OCBA procedure.

In this paper, we derive an alternative optimal budget allocation scheme for the fixed-budget robust R\&S  and then develop new additive robust OCBA procedures that are simpler but more efficient than existing ones. Specifically, we first adapt the additive PICS bound of fixed-precision procedures in \cite{fan2020distributionally} to a fixed-budget context, based on which we then intend to solve the optimal budget allocation that minimizes this PICS bound under a given sampling budget. Interestingly, we find that the budget allocation problem of robust R\&S problems involving $km$ scenarios may be translated into the traditional form for R\&S problems involving only $k+m-1$ scenario-defined alternatives, indicating an additive structure. These meaningful alternatives include the $m$ scenarios of the best alternative and the worse-case scenario of each non-best alternative. This result is beautiful from our viewpoint. It establishes a link between fixed-budget R\&S and robust R\&S. With this link, the optimal budget allocation scheme can be immediately obtained by adopting the traditional OCBA solution. This new solution further corroborates the aforementioned insight of \cite{gao2017robust} in a more clear and explicit way.

Moreover, we develop a new robust OCBA procedure based on the new optimal budget allocation solution. Given the above link between fixed-budget R\&S and robust R\&S, we are endowed the flexibility to  apply traditional OCBA procedures, e.g., the one of \cite{chen2000simulation}, to solve robust R\&S problems. By doing so, we identify an important issue that may significantly affect the performance of robust OCBA procedures. In each round, the OCBA procedures need to allocate a small batch of $\Delta \geq 1$ observations to the alternatives to meet the updated optimal allocation. This process may encounter a tricky numerical issue of dividing $\Delta$ to the alternatives with unmet budget needs; see Section 4.1.1 of \cite{chen2011stochastic} for a more detailed discussion. A common practice to avoid this is to allocate the observations to the most-starving alternative with the maximal budget gap. However, we observe that this most-starving strategy may lead to very poor performance in solving robust R\&S problems due to its greedy nature. To overcome this inefficiency, we propose to use a proportional allocation rule to determine the stage-wise budget allocation in the OCBA procedure tailored for solving robust R\&S problems. Numerical experiments show that doing so may significantly improve performance. The new OCBA procedure with the proportional allocation rule is named AR-OCBA procedure. Notice many other OCBA procedures with different allocation rules are readily easy to apply, but we choose to focus on this procedure.

Finally, we perform a comprehensive numerical study to understand the performance of the AR-OCBA procedure and compare it with the R-OCBA procedure of \cite{gao2017robust} and the R-UCB procedure of \cite{wan2023upper}. The numerical study produces several key findings. First, the AR-OCBA procedure displays superior performance over the R-OCBA and R-UCB procedures. Furthermore, this performance improvement may become more considerable as the problem scale grows, regardless of whether the problem scale being indexed by the number of alternatives $k$ or the total number of input distributions in the ambiguity set $m$. Second, as the total sampling budget grows for a fixed problem scale, the PCS of the AR-OCBA procedure may grow to 1, indicating its statistical consistency. This property is strongly desirable in the R\&S literature as it assures the value of using a larger sampling budget.  Moreover, the AR-OCBA procedure demonstrates a higher momentum towards consistency when given an increasing sampling budget compared to the R-OCBA and R-UCB procedures. Lastly, the sample size allocation analysis demonstrates the additive structure of the AR-OCBA procedure. Following our new additive optimal budget allocation scheme, the AR-OCBA procedure tends to allocate most of the sampling budget toward only $k+m-1$ scenarios among all $km$ scenarios, which may explain the observed advantages of the AR-OCBA procedure.

{We summarize the contributions of this paper as follows:\begin{itemize}
    \item We adopt an additive upper bound of the PICS to derive a new optimal budget allocation formulation for fixed-budget robust R\&S problems, further highlighting the additive structure inherent in robust R\&S problems.
    \item We identify a connection between the optimal budget allocation problem for robust R\&S and the OCBA approach for traditional R\&S. This connection allows us to derive the optimal budget allocation solution for robust R\&S problems effectively.
    \item  Based on the new optimal budget allocation solution, we propose the AR-OCBA procedure, which incorporates a proportional allocation rule to enhance its performance. This procedure is simpler than and demonstrated to outperform existing robust R\&S procedures.
\end{itemize}}

The rest of this paper is organized as follows. In Section \ref{sec: problem}, we introduce the fixed-budget robust R\&S problem and the additive rule of decomposing the PICS. Then, in Section \ref{sec: OCBA}, we derive a new optimal budget allocation solution under the additive decomposition, based on which we further design a sequential robust R\&S procedure in Section \ref{sec: procedures}. We include in Section \ref{sec: numerical} the numerical experiments and discuss interesting findings from the experiments. Lastly, we conclude the paper in Section \ref{sec: conclusion}.

\vspace{0.1cm}
\section{Problem Statement and Preliminaries}
\label{sec: problem}
\subsection{Fixed-Budget Robust Ranking \& Selection}
\label{subsec: problem_def}

Consider a finite collection of alternatives, denoted by \( S = \{s_{1}, s_{2}, \ldots, s_{k}\} \). Each alternative \( s \in S \) is associated with a performance measure \( g(s, \zeta) \), where \( \zeta \) represents the input variable that captures the uncertainty. The goal is to identify the best alternative, which is defined as the one with the smallest expected performance. Since the performance measure \( g(s, \zeta) \) is a black-box function, it is typically evaluated through simulations. To initiate these simulations, estimating the distribution of \( \zeta \) is necessary. However, there is inherent ambiguity in defining the exact distribution of \( \zeta \) due to limited input data, which leads to what is known as input uncertainty.

Following the approach in \cite{fan2013robust}, we model this uncertainty by introducing an ambiguity set \( \mathcal{P} = \{P_1, \ldots, P_m\} \), which represents a collection of \( m \) possible probability distributions that \( \zeta \) might follow. Adopting this robust perspective, we define the best alternative as the one that minimizes the worst-case expected performance across all distributions in the ambiguity set. Mathematically, this is expressed as
\begin{eqnarray}
\label{equation:minimax formulation}
    \min_{s\in S} \max_{P\in \mathcal{P}} \E_{P}[g(s,\zeta)].
\end{eqnarray}
This minimax framework, referred to as robust selection of the best (RSB) in \cite{fan2013robust, fan2020distributionally}, features a two-layer structure, each with its own objective. Specifically, in the inner layer, the goal is to identify the worst-case distribution from the ambiguity set \( \mathcal{P} \) that maximizes the mean performance for each alternative. In the outer layer, these worst-case performances are compared across all \( k \) alternatives, and the alternative with the smallest worst-case mean performance is selected as the best.

To clarify the analysis, we define a \textit{scenario} as the combination of an alternative \( s_i \in S \) and a probability distribution \( P_j \in \mathcal{P} \), denoted as scenario \( (i,j) \), where \( i = 1, 2, \ldots, k \) and \( j = 1, 2, \ldots, m \). Consequently, there are \( km \) scenarios in total. For each scenario \( (i,j) \), we may collect a series of independent and identically distributed observations through simulation, denoted as \( X_{ij,1}, X_{ij,2}, \dots \). The true mean and variance for scenario \( (i,j) \) are represented by \( \mu_{ij} = E_{P_j}[g(s_i, \zeta)] \) and \( \sigma_{ij}^{2} = \text{Var}_{P_j}[g(s_i, \zeta)] \), respectively. 
{
Following the convention in the R\&S literature, we assume that the observations of the scenarios are normally distributed, as summarized in the following assumption.
\begin{assumption}
\label{normal distribution}
    Observations of every scenario $(i,j)$ are {mutually independent} and normally distributed, i.e., ${X_{ij}} \sim N(\mu_{ij},\sigma_{ij}^2)$. 
\end{assumption}}

{Without loss of generality, we further assume that the true means of each alternative $i$ are in descending order, that is, $\mu_{i1} \geq \mu_{i2} \geq \ldots \geq \mu_{im}$ for each $i \in \{1, \cdots, k\}$. Furthermore, the worst-case means of the alternatives are also in descending order, satisfying $\mu_{k1} \geq \ldots \geq \mu_{21} >\mu_{11}$. } Consequently, 
alternative 1 is identified as the unique best alternative, with probability distribution 1 serving as its worst-case scenario, as defined by Equation \eqref{equation:minimax formulation}. For a clearer understanding,  Figure \ref{fig:example0} illustrates the mean configuration of an example problem where there are \( k = 3 \) alternatives and \( m = 3 \) distributions in the ambiguity set, resulting in a total of \( km = 9 \) scenarios. For each alternative \( i = 1, 2, 3 \), the mean performances follow the pattern \( \mu_{i1} > \mu_{i2} > \mu_{i3} \), meaning that scenario \( (i,1) \) represents the worst-case scenario for alternative \( i \). Among the worst-case scenarios \( (1,1), (2,1), \) and \( (3,1) \), the means satisfy \( \mu_{11} < \mu_{21} < \mu_{31} \). According to Equation \eqref{equation:minimax formulation}, alternative 1 is identified as the best alternative.

\begin{figure}
    \centering
    \includegraphics[width=0.37\linewidth]{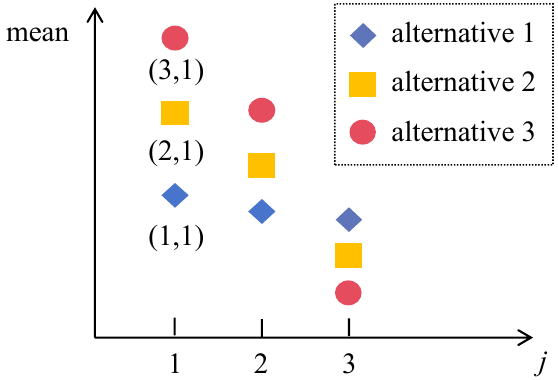}
    \caption{An example problem with $k=3$ and $m=3$}
    \label{fig:example0}
\end{figure}



In this paper, we address the problem in Equation \eqref{equation:minimax formulation} under a fixed-budget formulation. Let \( N \) denote the total  sampling budget, which is to be distributed across the \( km \) scenarios. The sample size for scenario \( (i,j) \) after $n$ observations has been allocated in total is denoted by \( n_{ij}(n) \), and \( \bar{X}_{ij}(n_{ij}) \) represents the sample mean for scenario \( (i,j) \) based on its corresponding sample size $n_{ij}$. Therefore, when the total  sampling budget \( N \) is exhausted, the terminal sample mean for scenario \( (i,j) \) is \( \bar{X}_{ij}(n_{ij}(N)) \). At this point, the worst-case scenario of each alternative will be selected based on the sample means, and the alternative \( \hat{i}^* \) with the minimal worst-case sample mean will be selected as the best alternative. This selection is based on comparing the sample means of all \( km \) scenarios, as expressed by
\begin{eqnarray}
  \label{eq: standard}
    \hat{i}^* = \underset{i \in \{1, \ldots,k\}}{\arg\min}\left(\underset{j \in \{1, \ldots,m\}}{\max} \bar{X}_{ij}(n_{ij}(N))\right).  
\end{eqnarray}


Due to the limited budget and inherent simulation noise, the selected alternative \( \hat{i}^* \) may not always correspond to the true best alternative. To evaluate the precision of the final selection, we use the Probability of Incorrect Selection (PICS). PICS measures the likelihood that the chosen best alternative \( \hat{i}^* \) is not equal to the actual best alternative, which is alternative 1.
The PICS of a fixed-budget robust R\&S procedure can be expressed as
\begin{eqnarray}
\label{equation:PCS}
         \notag \mbox{PICS} = \mbox{Pr}\left\{ \hat{i}^* \neq 1\right\} &=& \mbox{Pr}\left\{ \bigcup_{i=2}^k \left\{\underset{1\leq j \leq m}{\max} \bar{X}_{1j}(n_{1j}(N)) > \underset{1\leq j \leq m}{\max} \bar{X}_{ij}(n_{ij}(N))\right\} \right\}\\
        &=& \mbox{Pr}\left\{ \bigcup_{i=2}^k \bigcap_{j'=1}^m \bigcup_{j''=1}^m \left\{\bar{X}_{1j''}(n_{1j''}(N)) > \bar{X}_{ij'}(n_{ij'}(N)) \right\}\right\}.
\end{eqnarray}
Equation \eqref{equation:PCS} gives a formal definition of PICS, though it takes a complex and intractable form. 


In this work, our objective is to design an optimal budget allocation strategy that minimizes PICS, subject to the constraint of a finite sampling budget \( N \). Following the optimal-computing-budget-allocation (OCBA) approach, this challenge can be framed as a constrained optimization problem, formulated as follows:
\begin{eqnarray}
\label{problem formulation}
\begin{aligned}
\min_{n_{ij}} \quad & \mbox{PICS} & \\
\mbox{s.t.}\quad
&\sum_{i=1}^{k} \sum_{j=1}^{m} n_{ij}=N,&\\
&n_{ij}\geq 0, i=1,2,\ldots,k, j=1,2,\ldots,m.  &
\end{aligned}
\end{eqnarray}
For simplicity, the requirement that each $n_{ij}$ should be an integer is relaxed. Note that by setting \( m = 1 \), Problem \eqref{problem formulation} reduces to the traditional OCBA problem as described in \cite{chen2000simulation}, indicating that Problem \eqref{problem formulation} can be viewed as a generalization of the traditional OCBA problem. PICS serves as the objective to be optimized. However, as shown in Equation \eqref{equation:PCS}, the complexity and intractability of PICS present a significant challenge in solving the problem. Following the conventional approach in OCBA studies, we seek a simpler and more tractable approximation for PICS. In \cite{gao2017robust}, the PICS is decomposed into a more manageable form as
\begin{eqnarray}
\label{equation: Gao}
    \mbox{PICS} \leq\sum_{i=2}^k \sum_{j=1}^m \Pr \left\{ \bar{X}_{1j}(n_{1j}) > \bar{X}_{i1}(n_{i1}) \right\}.
\end{eqnarray}
The PICS bound in Equation \eqref{equation: Gao} is \textit{multiplicative} in that it involves $m(k-1)$ terms.  Having this bound for the PICS, \cite{gao2017robust} use it in Problem \eqref{problem formulation} as a surrogate objective, solve the problem when $N \rightarrow \infty$ and propose the R-OCBA procedure accordingly, which outperforms benchmark procedures in numerical experiments. However, we find it to be overly complex due to the multiplicative structure. In the next subsection, we introduce an additive upper bound from \cite{fan2020distributionally} to approximate the PICS.

\subsection{Additive Upper Bound of PICS}
\label{subsec: decomposition}

The multiplicative structure of the PICS bound in Equation \eqref{equation: Gao} makes Problem \eqref{problem formulation} still challenging to solve when using it as the objective. What we desire is a simpler decomposition approach that may simplify the form of the PICS bound. Coincidentally, we identify that \cite{fan2020distributionally} derive a such bound, although in a fixed-precision context. Specifically, in their Section 3.2, \cite{fan2020distributionally} derive that
\begin{eqnarray}
\label{equation: additive rule}
        \mbox{PICS}
        \notag &=& \Pr\left\{ \bigcup_{i=2}^k \left\{  \underset{1\leq j \leq m}{\max} \bar{X}_{1j} > \underset{1\leq j \leq m}{\max} \bar{X}_{ij} \right\}\right\}\\
        \notag  &\leq& \Pr\left\{ \bigcup_{i=2}^k \left\{ \underset{1\leq j \leq m}{\max} \bar{X}_{1j} > \underset{1\leq j \leq m}{\max} \bar{X}_{ij} \right\} \bigcup \bigcup_{j=2}^m \left\{ \bar{X}_{11} < \bar{X}_{1j} \right\} \right\}\\
        \notag  &=& \Pr\left\{\left\{ \bigcup_{i=2}^k \left\{ \underset{1\leq j \leq m}{\max} \bar{X}_{1j} > \underset{1\leq j \leq m}{\max} \bar{X}_{ij} \right\}\bigcap \left\{\bar{X}_{11} \geq \underset{1\leq j \leq m}{\max} \bar{X}_{1j} \right\}\right\} \bigcup \bigcup_{j=2}^m \left\{ \bar{X}_{11} < \bar{X}_{1j} \right\} \right\}\\
        \notag  &=& \Pr\left\{ \bigcup_{i=2}^k \left\{\bar{X}_{11} >  \underset{1\leq j \leq m}{\max} \bar{X}_{ij}\right\} \bigcup  \bigcup_{j=2}^m \left\{\bar{X}_{11} < \bar{X}_{1j}\right\} \right\}\\
       \notag   &\leq& \Pr\left\{ \bigcup_{i=2}^k \left\{ \bar{X}_{11} > \bar{X}_{i1} \right\}\right\} + \Pr \left\{ \bigcup_{j=2}^m \left\{ \bar{X}_{11} < \bar{X}_{1j}\right\} \right\}\\
     &\leq& \sum_{i=2}^k \Pr\left\{ \bar{X}_{11} >  \bar{X}_{i1}\right\} + \sum_{j=2}^m \Pr\left\{\bar{X}_{11} < \bar{X}_{1j}\right\},
\end{eqnarray}
where notations of sample sizes are omitted for clarity, and the second equality holds because $ A \cup B = (A \cap B^C) \cup B$ for events $A$ and $B$. This new PICS bound in Equation \eqref{equation: additive rule} may be referred to as \textit{additive} in that it involves only $k+m-2$ terms. 

Comparing the PICS upper bounds in Equation \eqref{equation: Gao} and Equation \eqref{equation: additive rule}, it is interesting to observe that, first, when $m=1$, the two bounds become identical and are equal to the traditional PICS bound used to solve the OCBA problem of generic R\&S problems \citep{chen2011stochastic}. Second, they both highlight the significance and relevance of only $k+m-1$ scenarios out of all $km$ scenarios including the $m$ scenarios of the best alternative and the worst-case scenario of every non-best alternative. This provides an important insight in deciding the optimal budget allocation. 

Despite these links, however, we also observe an essential difference between the two bounds. Equation \eqref{equation: Gao} involves pairwise comparisons between the $m$ scenarios of the best alternative and {the $k-1$ worse-case scenarios of the non-best alternatives}, resulting $m(k-1)$ terms. Equation \eqref{equation: additive rule}, however, takes a different viewpoint that centers around scenario (1,1). It implies that \textit{there are $k + m-2$ ``critical'' pairwise comparisons} \citep{fan2020distributionally}: scenario (1,1) against all other $m-1$ scenarios of the best alternative and scenario (1,1) against the worst-case scenario of every non-best alternative. Intuitively speaking, when $k$ and $m$ is relatively large, the additive bound in Equation \eqref{equation: additive rule} may be much tighter. In the next section, we will use this bound as the objective in Problem \eqref{problem formulation}  and  show that it may lead to a simple and clean solution to the problem.

It is worth noting that although we adopt the additive bound of the PICS from \cite{fan2020distributionally}, there are two key differences between our work and theirs. First, \cite{fan2020distributionally} apply this bound to design fixed-precision robust R\&S procedures, whereas we utilize it to address fixed-budget robust R\&S problems. Second, they believe that such a bound could not be applicable in sequential procedures and thus design only two-stage additive procedures, but as will we show in the next sections, we successfully design a well-performing sequential procedure with the additive bound.

\vspace{0.1cm}
\section{New OCBA Formulation for Robust R\&S}
\label{sec: OCBA}
In this section, we use the additive PICS bound in Equation \eqref{equation: additive rule} as the surrogate objective in Problem \eqref{problem formulation} and solve the asymptotically optimal budget allocation solution. Specifically, denote
\begin{equation*}
    f({\vec{n}}) = \sum_{i=2}^k \Pr\left\{\bar{X}_{11}(n_{11})> \bar{X}_{i1}(n_{i1}) \right\} + \sum_{j=2}^m \Pr\left\{ \bar{X}_{11}(n_{11})<\bar{X}_{1j}(n_{1j})\right\},
\end{equation*}
where $\vec{n}=(n_{11},\ldots,n_{21},\ldots,n_{k1},n_{12}\ldots,n_{1m})$ is a $km \times 1$ dimensional vector of sample sizes of all scenarios. {Under Assumption \ref{normal distribution}, the sample mean $\bar{X}_{ij}$ of scenario $(i,j)$ is also normally distributed, i.e., $\bar{X}_{ij} \sim N\left(\mu_{ij},\frac{\sigma_{ij}^2}{n_{ij}}\right)$, where $n_{ij}$ is the sample size of scenario $(i,j)$. 
Let $\Phi$ denote the cumulative density function of the standard normal distribution. By leveraging properties of the normal distribution, we have
    \begin{eqnarray}
    \label{equation: normalization0}
            f(\vec{n}) 
           \notag &=&  \sum_{i=2}^k \Pr\left\{\bar{X}_{i1}(n_{i1}) - \bar{X}_{11}(n_{11})<0  \right\} + \sum_{j=2}^m \Pr\left\{ \bar{X}_{11}(n_{11})-\bar{X}_{1j}(n_{1j})<0\right\}\\
           \notag &=& \sum_{i=2}^k \Phi \left( \frac{-(\mu_{i1}-\mu_{11})}{\sqrt{\frac{\sigma_{i1}^2}{n_{i1}} + \frac{\sigma_{11}^2}{n_{11}}}} \right) + \sum_{j=2}^m \Phi \left( \frac{-(\mu_{11}-\mu_{1j})}{\sqrt{\frac{\sigma_{11}^2}{n_{11}} + \frac{\sigma_{1j}^2}{n_{1j}}}} \right)\\
            &=& \sum_{i=2}^k \Phi \left( \frac{-\delta_{i1}}{\sqrt{\frac{\sigma_{i1}^2}{n_{i1}} + \frac{\sigma_{11}^2}{n_{11}}}} \right) + \sum_{j=2}^m \Phi \left( \frac{-\delta_{1j}}{\sqrt{\frac{\sigma_{11}^2}{n_{11}} + \frac{\sigma_{1j}^2}{n_{1j}}}} \right)
    \end{eqnarray}
where $\delta_{i1}:=\mu_{i1}-\mu_{11}$ and $\delta_{1j}:=\mu_{11}-\mu_{1j}$, representing the mean difference between scenario (1,1)  and other relevant scenarios. Notice that  $\delta_{i1} > 0$ and $\delta_{1j} \geq 0$  for $i=2,3,\ldots,k$ and $j=2,3,\ldots,m$.}

With $f(\vec{n})$, the budget allocation problem in Problem \eqref{problem formulation} becomes
\begin{eqnarray}
\label{problem formulation2}
\begin{aligned}
\min_{n_{ij}} \quad & f({\vec{n}}) & \\
\mbox{s.t.}\quad
&\sum_{i=1}^{k} \sum_{j=1}^{m} n_{ij}=N,&\\
&n_{ij}\geq 0, i=1,2,\ldots,k, j=1,2,\ldots,m.  &
\end{aligned}
\end{eqnarray}
For this problem, denote {$n_{ij}^*$} for $i=1, \cdots, k, j=1, \cdots, m$ as the optimal solution. Before solving the optimization problem, it can be observed that, among all $km$ scenarios, only $k+m-1$ scenarios are relevant to the problem, including the $m$ scenarios of the best alternative and the worst-case scenario of each non-best alternative. Furthermore, we can prove that
\begin{lemma}
\label{lemma: non-optimal}
    $n_{ij}^*=0$ when $j \neq 1$ and $i \neq 1$.
\end{lemma}
\begin{proof}{Proof.}
{In Equation \eqref{equation: normalization0}, the terms $ \frac{-\delta_{i1}}{\sqrt{{\sigma_{i1}^2}/{n_{i1}} + {\sigma_{11}^2}/{n_{11}}}}$ and $\frac{-\delta_{1j}}{\sqrt{{\sigma_{11}^2}/{n_{11}} + {\sigma_{1j}^2}/{n_{1j}}}}$ decrease as \(n_{11}\), \(n_{i1}\), or \(n_{1j}\) increases. Consequently, \(f(\vec{n})\) is monotonically decreasing with respect to \(n_{11}\), \(n_{i1}\), and \(n_{1j}\). Moreover, \(f(\vec{n})\) is independent of \(n_{ij}\) for \(i \neq 1\) and \(j \neq 1\). Therefore, the optimal solution to Problem \eqref{problem formulation2} must satisfy \(n_{ij}^* = 0\) when \(j \neq 1\) and \(i \neq 1\). 
Thus, the conclusion of interest is proved.}\hfill \Halmos
 
\end{proof}



Now we proceed to solve $n_{ij}^*$ for scenarios with {$i=1$, any $j$ or $i\neq1, j=1$}. For ease of presentation, we introduce a new index set $R=\{1,2,\ldots,k+m-1\}$.
Specifically speaking, $r=1$ corresponds to the worst-case scenario of the best alternative, i.e., scenario $(1,1)$. $r=2,\ldots,k$ corresponds to the worst-case scenarios of non-best alternatives, i.e., scenario $(i,1)$ for $i \neq 1$. $r=k+1,\ldots,k+m-1$ corresponds to the non-worst-case scenarios of the best alternative, i.e., scenario $(1,j)$ for $j \neq 1$. 
See Figure \ref{fig:new index} for an intuitive understanding. 
\begin{figure}[htbp]
\centering\includegraphics[width=0.6\linewidth]{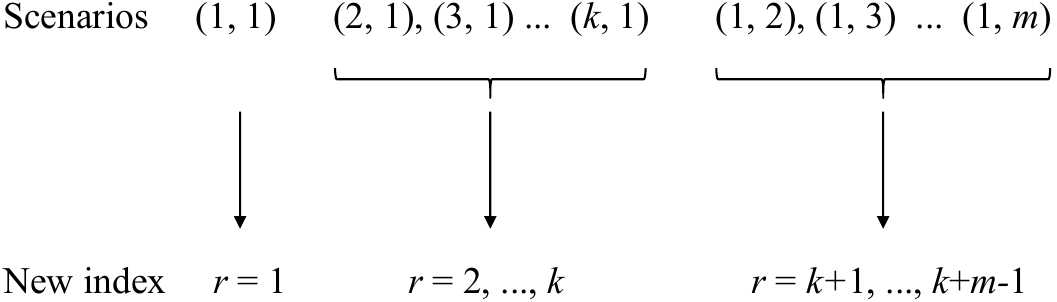}
    \caption{Formalization of the new index set}
    \label{fig:new index}
\end{figure}

Equipped with the new index set, $f(\vec{n})$ in Equation \eqref{equation: normalization0} can be rewritten as
\begin{eqnarray}
\label{equation: final of f}
    f(\vec{n}) = \sum_{r=2}^{k+m-1} \Phi \left( \frac{-\delta_r}{\sqrt{\frac{\sigma_1^2}{n_1}+\frac{\sigma_r^2}{n_r}}} \right)
\end{eqnarray}
where $\delta_r = |\mu_{1} - \mu_{r}|$ is the absolute mean difference between two scenarios in the new index set, and $\sigma_r^2$ and $n_r$ are variance and sample size of the corresponding scenario. 
Since $f(\vec{n})$ is exclusively determined by the sample sizes of the $k+m-1$ relevant scenarios, we can reformulate Problem \eqref{problem formulation2} as
\begin{eqnarray}
\label{new formulation}
\begin{aligned}
\max_{n_{r}} \quad & 1-\sum_{r=2}^{k+m-1} \Phi \left( \frac{-\delta_r}{\sqrt{\frac{\sigma_1^2}{n_1}+\frac{\sigma_r^2}{n_r}}}\right) & \\
\mbox{s.t.}\quad
&\sum_{r=1}^{k+m-1}n_{r}=N,&\\
&n_{r}\geq 0,r=1,2,\ldots,k+m-1.  &
\end{aligned}
\end{eqnarray}
Interestingly, we find this new optimization problem is equivalent to the OCBA problem for a traditional R\&S problem with $k+m-1$ alternatives and alternative $r=1$ being the best alternative. 
This finding shows that with the additive PICS lower bound in Equation \eqref{equation: additive rule}, a two-layer robust R\&S problem with $km$ scenarios can be treated as a single-layer R\&S problem with only $k+m-1$ alternatives under the OCBA framework. 
Such transformation has the advantage of significantly reducing the problem scale.  For example, if $k=m=20$, it turns a robust R\&S problem with $km=400$ scenarios to a single-layer R\&S problem with only $k+m-1=39$ alternatives. The problem scale reduces by $90.25\%$ in this example. Conceivably, the larger $k$ or $m$ is, the more significant the reduction becomes. Also, it allows for the flexible application of various traditional OCBA procedures to solve robust R\&S problems. Later in Section \ref{sec: procedures} we will make use of this transformation to develop a new fixed-budget robust R\&S procedure.


In light of this equivalence, 
we may immediately obtain the optimal solution to the new problem by adapting traditional OCBA solutions from \cite{chen2000simulation}. Then we obtain Lemma \ref{lemma: solution}.

\begin{lemma}
\label{lemma: solution}
    As $N \rightarrow \infty$, the optimal solution to Problem \eqref{new formulation} satisfies
\begin{eqnarray}
\label{equation: solution}
    \begin{aligned}
    \frac{n_r^*}{n_{s}^*} &= \left( \frac{\sigma_r/\delta_r}{\sigma_{s}/\delta_{s}}\right)^2 \text{ for all } r,s = 2,\ldots,k+m-1,\\
    n_1^* &= \sigma_1 \sqrt{\sum_{r=2}^{k+m-1} \frac{n_r^{*2}}{\sigma_r^2}},
\end{aligned}
\end{eqnarray}
where $n_r^*$ is the optimal sample size of alternative $r$. 
\end{lemma}

Combining Lemma \ref{lemma: non-optimal} and Lemma \ref{lemma: solution}, we may obtain the optimal solution to Problem \eqref{problem formulation2}, which is stated in the following theorem.

\begin{theorem}
\label{theorem}
\label{thm: OCBA}
Suppose that  Assumption \ref{normal distribution} holds. Given a total sampling budget $N$ to be allocated to $km$ scenarios, as $N \rightarrow \infty$, 
the optimal solution to Problem \eqref{problem formulation2} satisfies
\begin{eqnarray}
\label{equation: solution2}
    \begin{aligned}
n_{ij}^* &=0 \text{ for all } i = 2, \cdots, k, j=2\cdots, m,\\
    \frac{n_{p1}^*}{n_{q1}^*} &= \left( \frac{\sigma_{p1}/(\mu_{p1}-\mu_{11})}{\sigma_{q1}/(\mu_{q1}-\mu_{11})}\right)^2 \text{ for all } p, q=2, \cdots, k,\\
\frac{n_{1s}^*}{n_{1t}^*} &= \left( \frac{\sigma_{1s}/(\mu_{11}-\mu_{1s})}{\sigma_{1t}/(\mu_{11}-\mu_{1t})}\right)^2 \text{ for all } s, t=2,\cdots, m,\\
\frac{n_{i1}^*}{n_{1j}^*} &= \left( \frac{\sigma_{i1}/(\mu_{i1}-\mu_{11})}{\sigma_{1j}/(\mu_{11}-\mu_{1j})}\right)^2 \text{ for all } i = 2, \cdots, k, j=2\cdots, m,\\
n_{11}^* &= \sigma_{11} \sqrt{\sum_{j=2}^{m}\frac{n_{1j}^{*2}}{\sigma_{1j}^2} + \sum_{i=2}^{k}\frac{n_{i1}^{*2}}{\sigma_{i1}^2} }.
\end{aligned}
\end{eqnarray}
\end{theorem}



In Theorem \ref{theorem}, the optimal allocation of observations \( n_{ij}^* \) across various scenarios \( (i, j) \) is determined by several key relationships, which may help prioritize which scenarios should receive attention in the sampling process. These relationships convey two critical points of the optimal budget allocation scheme. First, in the optimal solution,  most of the scenarios receive no observations at all, highlighting the additive structure of robust R\&S. This dramatically reduces the number of scenarios needing attention, focusing budget allocation on a subset of scenarios that includes \( i = 1 \) or \( j = 1 \). Second, scenario \( (1, 1) \) serves as a reference point for the allocation of observations. It plays a pivotal role in the entire allocation scheme, influencing how the observations are distributed. 


\vspace{0.1cm}
\section{New Sequential OCBA Procedure for Robust R\&S}
\label{sec: procedures}

Optimally allocating the total sampling budget, as described in Theorem \ref{thm: OCBA}, assumes that the means and variances of the alternatives are known in advance, which is unrealistic in practical applications. To address such limitation, the OCBA literature has developed sequential procedures. These procedures iteratively plug in sample estimates of the unknown parameters into balancing equations as in Equation \eqref{equation: solution2}, allowing the budget allocation to adapt as more data becomes available. At each step, a new small batch of observations is allocated based on the updated estimates, and this process continues until the total sampling budget is exhausted \citep{chen2000simulation, gao2017robust}. Given the equivalence between the OCBA problem for robust R\&S with $km$ scenarios (Problem \ref{problem formulation2}) and the OCBA problem for traditional R\&S with $k+m-1$ alternatives (Problem \ref{new formulation}), we may adapt existing sequential OCBA algorithms to solve robust R\&S problems. In this section, we develop an efficient sequential OCBA procedure specifically tailored for solving robust R\&S problems.

\subsection{Procedure Design}

Following the standard OCBA approach \citep{chen2011stochastic}, we propose a meta OCBA procedure for solving robust R\&S problems, detailed in Procedure \ref{algo1}. The procedure consists of two phases. In the initialization phase, it takes $n_0 \geq 2$ observations from each scenario $(i,j)$ to estimate the sample mean $\bar{X}_{ij}(n_0)$ and sample variance $\hat{\sigma}_{ij}^2(n_0)$. Afterward, the procedure transitions into the sequential allocation phase. At each stage, it first identifies the worst-case scenario for each alternative $i$, based on the current sample information, denoted as $\hat{1}_i$. Among these worst-case scenarios, it then identifies the current best alternative $\hat{b} = \argmin_{i=1,\cdots,k} \bar{X}_{i\hat{1}_i}$.  Using the identified worst-case scenarios for each alternative and the remaining scenarios for the current best alternative, the procedure increases the sampling budget by $\Delta \geq 1$ and recalculates the optimal budget allocation. It then allocates a small batch of $\Delta$ observations to the $k+m-1$ relevant scenarios to meet the updated allocation. When a scenario receives new observations, its sample mean, sample variance, and sample size are updated accordingly. Upon exhausting the total sampling budget $N$, the procedure concludes by returning the current best alternative as the final best alternative.

\begin{algorithm}[htbp]
  \caption{Meta OCBA Procedure for Robust R\&S}
  \begin{algorithmic}[1]
    \Require
        $km$ scenarios of $k$ alternatives, total sampling budget $N$,  
      initialization sample size $n_0 \geq 2$, and batch size $\Delta \geq 1$.
       \State Take $n_0$ observations from each scenario $(i,j)$ and calculate sample mean $\bar{X}_{ij}(n_0)$ and sample  variance $\hat{\sigma}_{ij}^2(n_0)$.
       \State Set $n_{ij} \gets n_0$ for each scenario $(i,j)$, $N_{used} \gets kmn_0$ and $t \gets 0$;
       \While {$N_{used}+\Delta < N$}
       \State Find $\hat 1_i  \gets \argmax_{j=1, \cdots, m} \bar{X}_{ij}(n_{ij})$ for each alternative $i = 1,2,\ldots,k$;
       \State Identify the current best alternative $\hat{b} \gets \argmin_{i=1, \cdots, k} \bar{X}_{i\hat 1_i}(n_{i\hat 1_i})$;
       \State {Form the new alternative set $R$ with $k+m-1$ scenarios $(1, \hat 1_1), \cdots, (i, \hat 1_i), (\hat{b},1), \cdots, (\hat{b}, m)$ where $i = 1, \ldots, k$ and $i \neq \hat{b}$;}
              \State Let $t \gets t+1, N_{used} \gets N_{used} + \Delta$ and calculate the new optimal budget allocation $\{n_r^t\}_{r=1, \cdots, |R|}$ within $R$ according to Equation \eqref{equation: solution} for the sampling budget $N_{used}$;
       \State \textit{Allocate $\Delta$ observations among scenarios in $R$ according to $\{n_r^t\}_{r=1, \cdots, |R|}$}; \label{line allocation}
\State Update the sample size, sample mean and sample variance for each scenario in $R$;
    \EndWhile
    \State Select alternative $\argmin_{i=1, \cdots, k}\max_{j=1, \cdots, m} \bar{X}_{ij}(n_{ij})$. 
  \end{algorithmic}
  \label{algo1}
\end{algorithm}


Procedure \ref{algo1} outlines the framework of our OCBA procedure for solving robust R\&S problems, which leverages the connection between Problem \eqref{problem formulation2} and Problem \eqref{new formulation} as discussed in the previous section. Notice that in the sequential allocation phase, a stage-wise allocation rule is required to determine how the incremental $\Delta$ observations should be allocated to the $k+m-1$ scenarios in the set $R$ (see Line \ref{line allocation}). This allocation rule may significantly impact the performance of the procedure. In the following subsection, we explore proper allocation rules for $\Delta$.

\subsection{Stage-wise Budget Allocation Rules}
\label{subsec: rules}
In each sequential allocation stage of OCBA procedures, the incremental budget  $\Delta$ should  ideally be allocated so that the updated sample sizes $n_r$ closely approximate the updated optimal sample size $n_{r}^{t}$ for each alternative $r$ (here we use the terminology for traditional R\&S problems since we focus on the $k+m-1$ scenarios in the set $R$).   To meet such condition, traditional OCBA procedures, such as the one of \cite{chen2000simulation}, allocate additional $\max(0,n_{r}^{t} - n_{r})$ observations to each alternative $r$, i.e., allocate as many observations as required by each alternative. However, this satisficing rule  comes with a problem: the total number of allocated observations 
$\sum_{r=1}^{k+m-1}\max(0,n_{r}^{t} - n_{r})$
may be significantly larger than $\Delta$, particularly when the total number of alternatives $k+m-1$ is large, causing inefficient use of the sampling budget. To overcome this issue, a common practice is to use the {most-starving rule}. 

The \textit{most-starving} rule allocates the incremental $\Delta$ observations to the alternative $r^*$ that requires the most additional observations, i.e., the alternative with the largest sample size gap:
\begin{equation*}
    r^*={\argmax_{r=1, \cdots, k+m-1}\{0,n_{r}^{t}-n_{r}\}}. 
\end{equation*}
This rule is simple to implement, but when alternative $r^*$ is actually a scenario with $i\neq 1$ and $j\neq 1$, it may receive more observations than necessary (see Theorem \ref{theorem}), leading to the inefficient use of the limited budget. Furthermore, the other scenarios that also require additional observations might be neglected, resulting in insufficient exploration. A common, special case of this rule occurs when $\Delta=1$, where only one observation is allocated to the scenario with the maximum sample size gap at each stage. However, this may possibly lead to a slow allocation process.

In this paper, we propose a new \textit{proportional allocation rule}. Under this rule, the number of observations allocated to each alternative $r$ in the set $R$ at stage $t$ is given by
\begin{eqnarray}
\label{equation: pro}
    \left\lceil \frac{\max\{0,n_{r}^{t}-n_r\} \Delta}{\sum_{r}^{k+m-1} \max\{0,n_{r}^{t}-n_r\}}   \right\rceil.
\end{eqnarray}
This formula ensures that the incremental $\Delta$ observations will be distributed proportionally among the scenarios that require additional observations, i.e., satisfy $n_{r}^{t} - n_r > 0$. If $n_{r}^{t} - n_r \leq 0$, the scenario is considered to need no further observations. Unlike the most-starving rule, the proportional rule takes into account all scenarios requiring additional observations and allocates the budget in proportion to the extent of their need.

\begin{algorithm}[htbp]
  \caption{AR-OCBA Procedure}
  \begin{algorithmic}[1]
    \Require
        $km$ scenarios of $k$ alternatives, total sampling budget $N$,  
      initialization sample size $n_0 \geq 2$, and batch size $\Delta \geq 1$.
       \State Take $n_0$ observations from each scenario $(i,j)$ and calculate sample mean $\bar{X}_{ij}(n_0)$ and sample  variance $\hat{\sigma}_{ij}^2(n_0)$.
       \State Set $n_{ij} \gets n_0$ for each scenario $(i,j)$, $N_{used} \gets kmn_0$ and $t \gets 0$;
       \While {$N_{used}+\Delta < N$}
             \State Lines 4-6 of Procedure \ref{algo1};
       \State Let $t \gets t+1$ and calculate the new optimal budget allocation $\{n_r^t\}_{r=1, \cdots, |R|}$ within $R$ according to Equation \eqref{equation: solution} for the sampling budget $N' \gets N_{used} + \Delta$;
       \State Take $\left\lceil \frac{\max\{0,n_{r}^{t}-n_r\} \Delta} {\sum_{r}^{k+m-1} \max\{0,n_{r}^{t}-n_r\}}   \right\rceil$ observations from each scenario $r$ in the set $R$ and let $N_{used} \gets N_{used} + \sum_{i=1}^{k+m-1} \left\lceil \frac{\max\{0,n_{r}^{t}-n_r\} \Delta} {\sum_{r}^{k+m-1} \max\{0,n_{r}^{t}-n_r\}}   \right\rceil$;
       \State Update the sample size, sample mean and sample variance for each scenario in $R$;
    \EndWhile
    \State Select alternative $\argmin_{i=1, \cdots, k}\max_{j=1, \cdots, m} \bar{X}_{ij}(n_{ij})$. 
  \end{algorithmic}
  \label{AR-OCBA}
\end{algorithm}

We compare these two allocation rules in Section \ref{susec: compare_rules} through numerical experiments with varying parameter settings. Our results show that the proportional rule may significantly outperform the most-starving rule in minimizing the PICS. Therefore, we adopt the proportional rule in our new OCBA procedure, which we refer to as the AR-OCBA procedure. Procedure \ref{AR-OCBA} provides a detailed description of the AR-OCBA procedure. 


\vspace{0.1cm}
\section{Numerical Experiments}
\label{sec: numerical}
In this section, we conduct extensive numerical experiments to understand the performance of our AR-OCBA procedure and {compare it with that of existing fixed-budget robust R\&S procedures.} Specifically, in Section \ref{susec: compare_rules}, we study the performance of the AR-OCBA procedure with different stage-wise budget allocation rules that are discussed in Section \ref{subsec: rules}. We then compare the procedure with {existing robust R\&S procedures} in Section \ref{subsec: comparison}. Moreover,  we conduct a sensitivity analysis for several procedure parameters in Section \ref{subsec: parameters}. Lastly, we investigate the sampling budget allocation behaviors of the AR-OCBA procedure in Section \ref{subsec: allocation}.

Throughout the first three subsections, we follow \cite{gao2017robust} and \cite{fan2020distributionally} to consider the following monotone configuration for the mean values (MM) of all $km$ scenarios:
\begin{eqnarray*}
            \left[\mu_{i j}\right]_{k \times m} = \big[ 0.5 \times i - 0.2 \times j-1\big]_{1\leq i \leq k, 1\leq j \leq m}.
\end{eqnarray*}
For the variances, we consider three types of configurations:
\begin{itemize}
    \item Configuration with constant variances (CV): $\sigma_{ij}^2 = 16^2$ for all $i, j$.
    \item Configuration with increasing variances (IV): $\sigma_{ij}^2 = \left(12+\sqrt{0.2i+j}\right)^2$ for all $i, j$.
    \item Configuration with decreasing variances (DV): $\sigma_{ij}^2 = \left(12 + \frac{1} {0.2i + j}\right)^2$ for all $i, j$.
\end{itemize}
Combining the mean configuration and the variance configurations, we may obtain three problem configurations: {MM}-CV, {MM}-IV, and {MM}-DV. We will base our analysis in the first four subsections on these problem configurations. 

When implementing the AR-OCBA procedure and the R-OCBA procedure, we need to specify the value of $n_0$ and batch size $\Delta$. By default, we let $n_0=20$ and $\Delta=20$. We will explore the robustness of the AR-OCBA procedure to these parameters in Section \ref{subsec: parameters}. In every experiment, we may vary the values of $k, m$ or the total sampling budget $N$ to explore how the probability of correct selection (PCS) of the procedures changes accordingly. Furthermore, in all experiments we estimate the PCS of a procedure in solving a specific problem based on 4000 independent macro replications.

\subsection{Comparing Stage-wise Allocation Rules}
\label{susec: compare_rules}
In this subsection, we compare the performance of the AR-OCBA procedure with different stage-wise allocation rules. Specifically, we compare the following procedures:
\begin{itemize}
    \item AR-OCBA procedure: the procedure which is described in Procedure \ref{AR-OCBA}.
    \item AR-OCBA procedure (Starving): the procedure described in Procedure \ref{AR-OCBA} but using the most starving strategy in each round to allocate the budget $\Delta$. 
        \item Equal-allocation (EA) procedure: the procedure allocates the total sampling budget equally to all the $km$ scenarios and selects the alternative according to Equation \eqref{eq: standard}. 
\end{itemize}
The EA procedure is included  as a benchmark. For each of the three procedures, we estimate the PCS under the three problem configurations MM-CV, MM-IV and MM-DV; for each configuration, we let $k=20,m=5$ and set the total sampling budget $N=(n_0+c)km$. The PCS of every procedure against different $c$ for each problem instance is plotted in Figure \ref{fig:allocation_rule}.

\begin{figure}[htbp]
    \centering
    \includegraphics[width=0.88\linewidth]{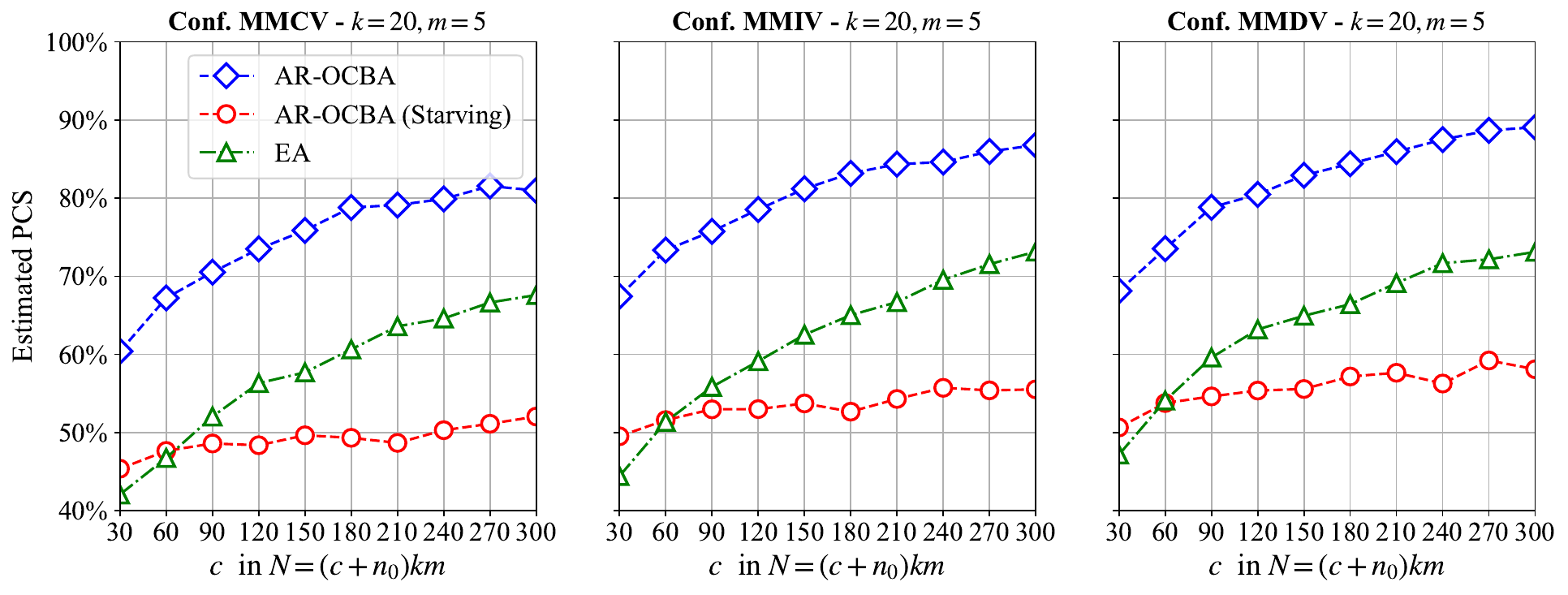}
    \caption{A comparison among the PCS of the AR-OCBA procedure with different stage-wise allocation rules}
    \label{fig:allocation_rule}
\end{figure}

From Figure \ref{fig:allocation_rule} we may observe the clear advantage of using the proportional allocation rule in the AR-OCBA procedure. For all problem configurations, the PCS of the AR-OCBA procedure is much higher than that of the procedure using the most starving allocation rule. Interestingly, the PCS of the latter may be even lower than that of the naive EA procedure, suggesting that using the most starving allocation rule in the AR-OCBA procedure may not be a reasonable choice. This result may be attributed to the fact that the most starving allocation rule is too ``greedy'' in allocating the sampling budget. Instead, the proportional allocation rule may well balance the sampling needs of the alternatives at each stage. Therefore, we adhere to the proportional rule in the design of the AR-OCBA procedure. Furthermore, with this sufficient exploration, the AR-OCBA procedure {appears to achieve the consistency as the total sampling budget grows, as we will see in the next subsection; we will discuss the consistency issue later in Section \ref{subsec: allocation}}.

\subsection{{Comparing Fixed-Budget Robust R\&S Procedures}}
\label{subsec: comparison}
In this subsection, we compare our AR-OCBA procedure with the R-OCBA procedure \citep{gao2017robust} and the R-UCB procedure \citep{wan2023upper}. We first compare them using the three problem configurations. 
In this comparison, we consider multiple values of $k$ and $m$ to represent different problem scales. Specifically, we set \(k = 20\), \(100\), and \(200\) while fixing \(m = 5\). Additionally, we vary \(m = 5\), \(20\), and \(50\) while keeping \(k = 20\). In each setting, we compare the PCS of the three procedures under various values of $c$ in the total sampling budget \(N=(n_0+c)km\). The PCS curves of the three procedures against $c$ for different $k$ are plotted in Figure \ref{fig:compare_k}, while those for different $m$ are plotted in Figure \ref{fig:compare_m}.

\begin{figure}[htbp]
    \centering
    \includegraphics[width=0.88\linewidth]{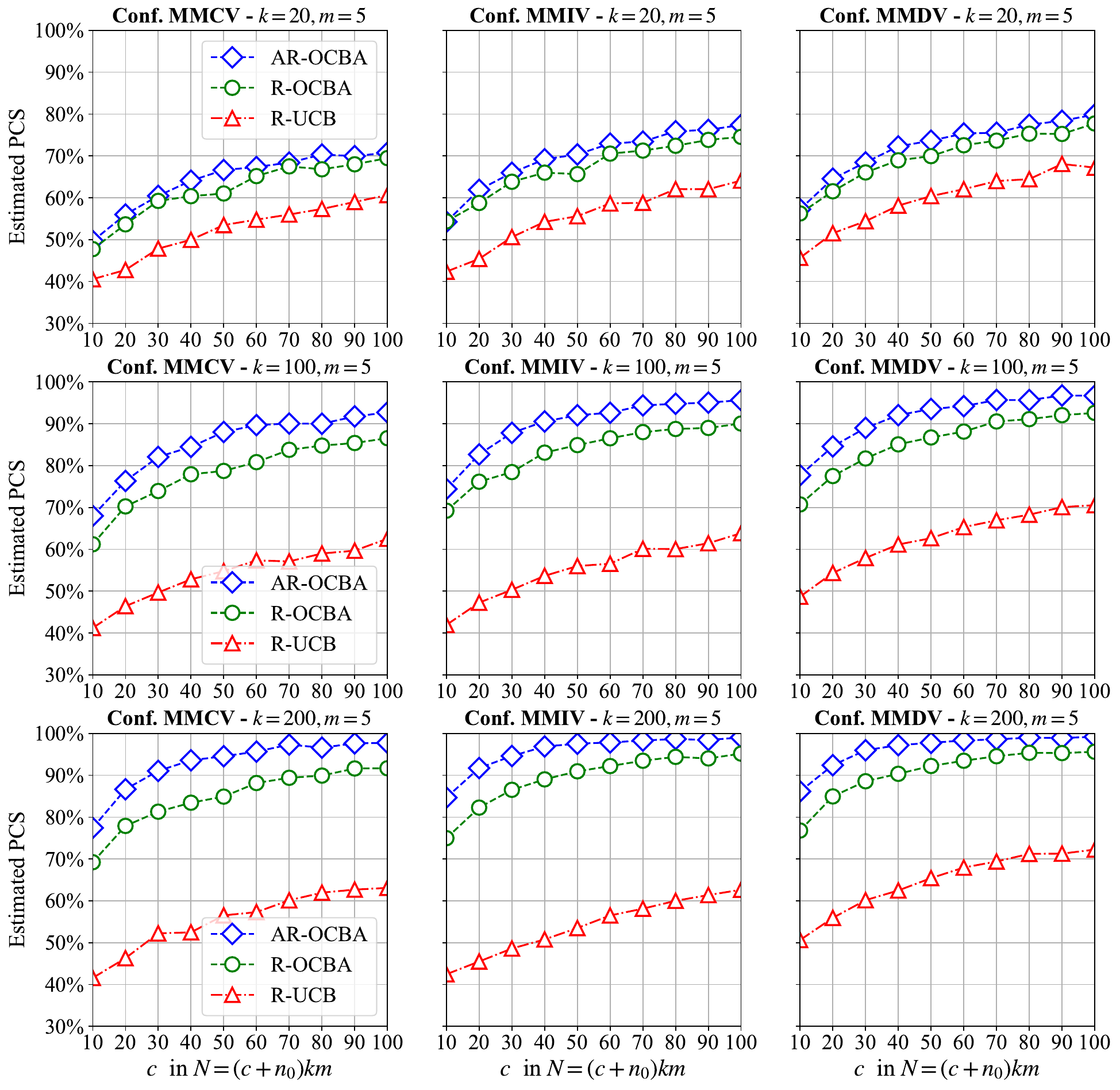}
    \caption{{A comparison among the PCS of the AR-OCBA, R-OCBA, and R-UCB procedures under different $\mathbf{k}$}}
    \label{fig:compare_k}
\end{figure}

\begin{figure}[htbp]
    \centering
    \includegraphics[width=0.88\linewidth]{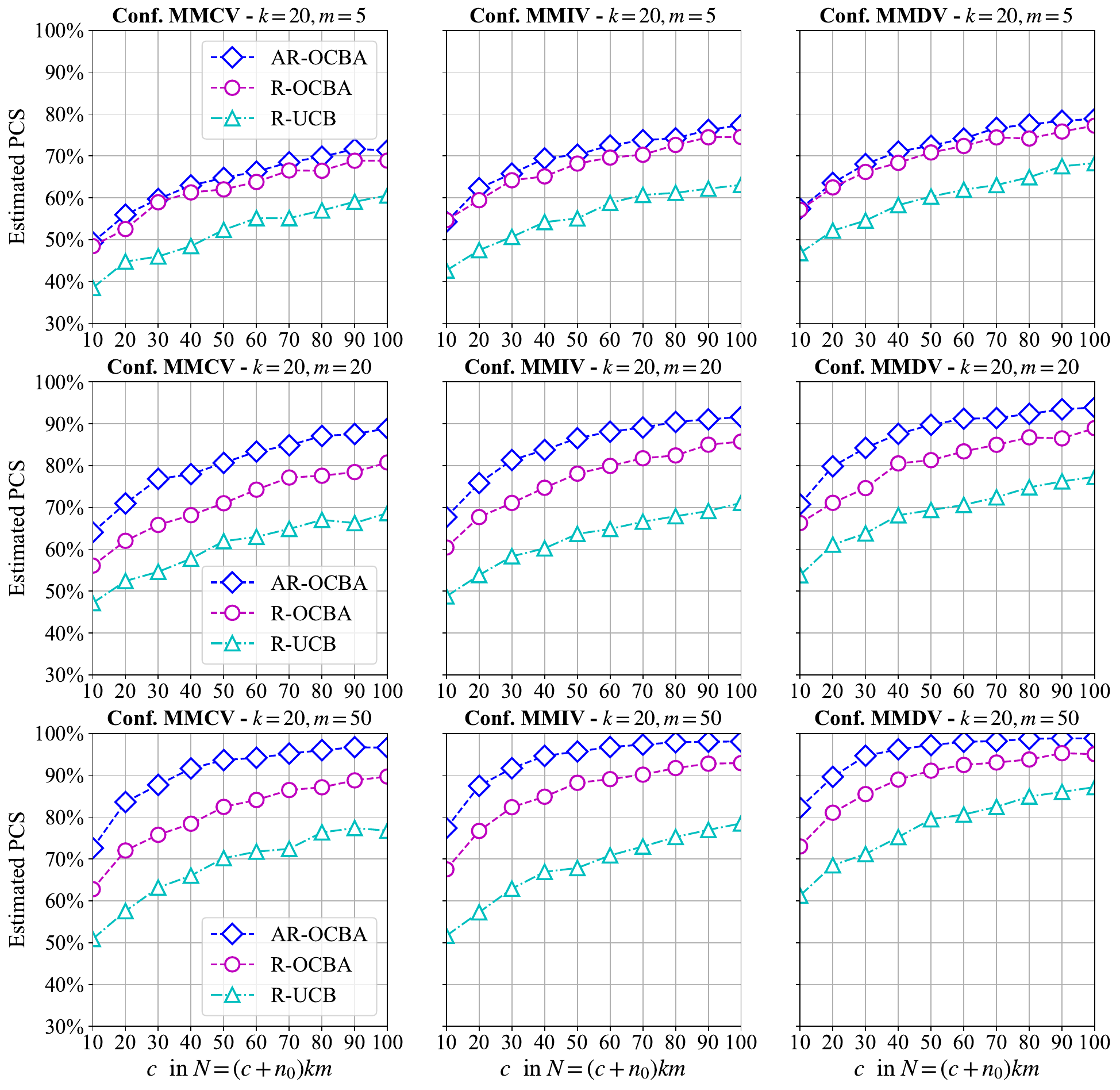}
    \caption{{A comparison among the PCS of the AR-OCBA, R-OCBA, and R-UCB procedures under different $\mathbf{m}$}}
    \label{fig:compare_m}
\end{figure}

As shown in Figure \ref{fig:compare_k}, AR-OCBA consistently demonstrates superior performance compared to R-OCBA {and R-UCB} across all problem configurations, all values of \(k\) and all budget levels. Despite the much simpler structure of the AR-OCBA procedure, it achieves slightly higher PCS than R-OCBA {and significantly higher PCS than R-UCB} at all tested budget levels when \(k = 20\) and \(m = 5\). Interestingly, when increasing \(k\) to 100 and 200, the advantage of AR-OCBA becomes more pronounced, as seen in the widening PCS {gaps between AR-OCBA and the other two procedures.} This result indicates that the AR-OCBA procedure effectively scales with $k$ and the procedure is more well-suited for scenarios with a large number of alternatives. 

The PCS comparison in Figure \ref{fig:compare_m} leads to similar findings in the above for $m$. In the figure, the AR-OCBA procedure maintains its performance edge as the number of distributions in the ambiguity set $m$ increases from \(5\) to \(50\). For example, under the MM-CV configuration, AR-OCBA continues to outperform R-OCBA {and R-UCB} across all budget levels and its performance advantage over R-OCBA {and R-UCB} becomes more evident as $m$ increases, reflecting its robustness against increased input uncertainty. This pattern persists under both MM-IV and MM-DV configurations, reinforcing its superiority in scenarios with a large number of distributions in the ambiguity set. 

{Furthermore, in Figure \ref{fig:compare_k} and Figure \ref{fig:compare_m},  the AR-OCBA procedure appears to achieve the consistency. Under both MMIV and MMDV configurations with $k=200$ or $m=50$, the PCS of the AR-OCBA procedure may become very close to 1 as the total sampling budget $N$ grows, which implies that using the proportional allocation rule in the procedure may provide sufficient exploration for these configurations. Later in Section \ref{subsec: allocation}, we will analyze the sampling budget allocation behavior of AR-OCBA and discuss the consistency issue in-depth.}

{
We now compare the performance of the three procedures in solving a practical R\&S problem. We adopt an \((s, S)\) inventory problem from the SimOpt library. Detailed descriptions of the problem settings can be found in \href{https://simopt.readthedocs.io/en/latest/sscont.html}{https://simopt.readthedocs.io/en/latest/sscont.html}. In this problem, the decision variables \(s\) and \(S\) represent the reorder point and the target inventory level, respectively. When the inventory position falls below \(s\), an order is placed to raise the inventory level to \(S\).  
In our example, the goal is to find the optimal \((s, S)\) pair that minimizes the average cost over a 500-day horizon, and  we follow the recommended parameter settings to consider the case where \(s \in [700, 1000]\) and \(S \in [1500, 2000]\). To formulate the robust R\&S problem, we discretize the decision space as \(s \in \{700, 725, 750, 775, 800, 825, 850, 875, 900, 925, 950, 975, 1000\}\) and \(S \in \{1500, 1550, 1600, 1650, 1700, 1750, 1800, 1850, 1900, 1950, 2000\}\), resulting in $k=143$ alternatives. The input uncertainty arises from the estimation of the demand distribution, and we consider an ambiguity set consisting of $m=9$ exponential demand distributions with means \(\{40, 45, 50, 55, 60, 65, 70, 75, 80\}\). In total, this leads to \(9 \times 143 = 1287\) scenarios.  To identify the worst-case scenarios and the best alternative, we estimate the mean performance of each scenario based on 10,000 independent macro replications.}

{For the comparison, we set the initial sample size \(n_0 = 10\) and use \(\Delta = 10\) for both the AR-OCBA and R-OCBA procedures. As before, the total sampling budget \(N\) is set as \(N = (n_0 + c)km\), where \(c\) varies from 10 to 40. For each \(c\), we estimate the PCS of each procedure using 1,000 independent macro replications. The PCS results, illustrated in Figure \ref{fig:example}, demonstrate that the AR-OCBA procedure consistently outperforms the other two procedures. This highlights the practical advantages of the AR-OCBA procedure in solving robust R\&S problems.
}

\begin{figure}[htbp]
    \centering
    \includegraphics[width=0.45\linewidth]{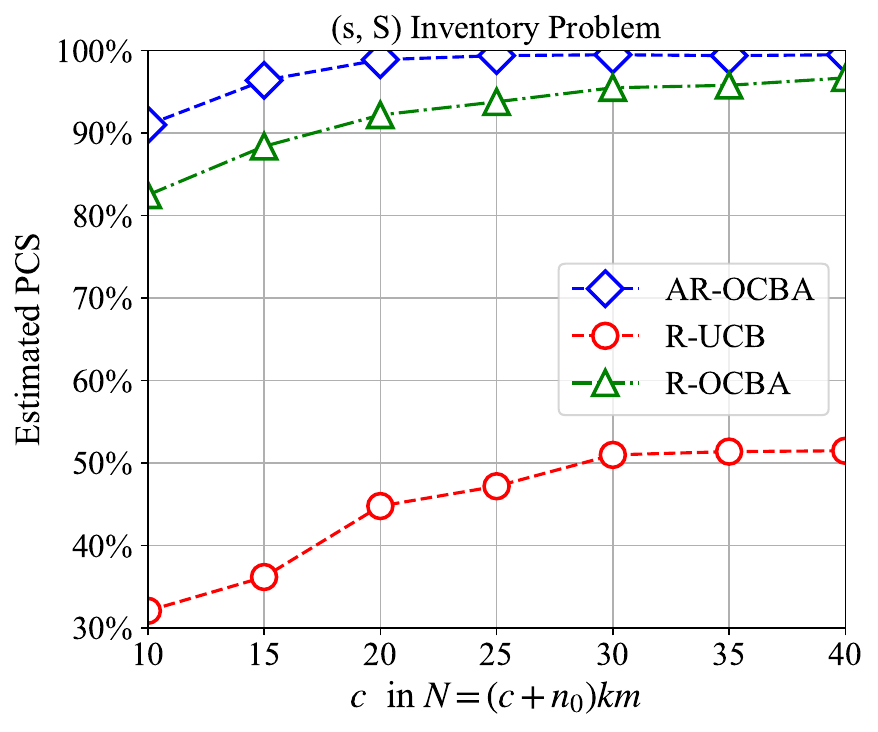}
    \caption{{A comparison among the PCS of the AR-OCBA, R-UCB, and R-OCBA procedures for the (s, S) inventory problem}}
    \label{fig:example}
\end{figure}

\subsection{Robustness Check for Procedure Parameters}
{In this subsection, we explore the impact of the choices of $n_0$ and $\Delta$ on the performance of the AR-OCBA procedure. We fix $k=20$, $m=5$ and let the total sampling budget $N=(c+n_0)km=50km$. Under this setting, we vary $n_0$ from {$3$ to $36$} while letting {$\Delta=10$} to test the impact of $n_0$. We then vary $\Delta$ from {2 to 24} while fixing {$n_0=10$} to test the impact of $\Delta$. The PCS curves of the AR-OCBA procedure against different $n_0$ and $\Delta$ are plotted in Figure \ref{fig:robustness}.}

 \label{subsec: parameters}
\begin{figure}[htbp]
    \centering
    \includegraphics[width=0.88\linewidth]{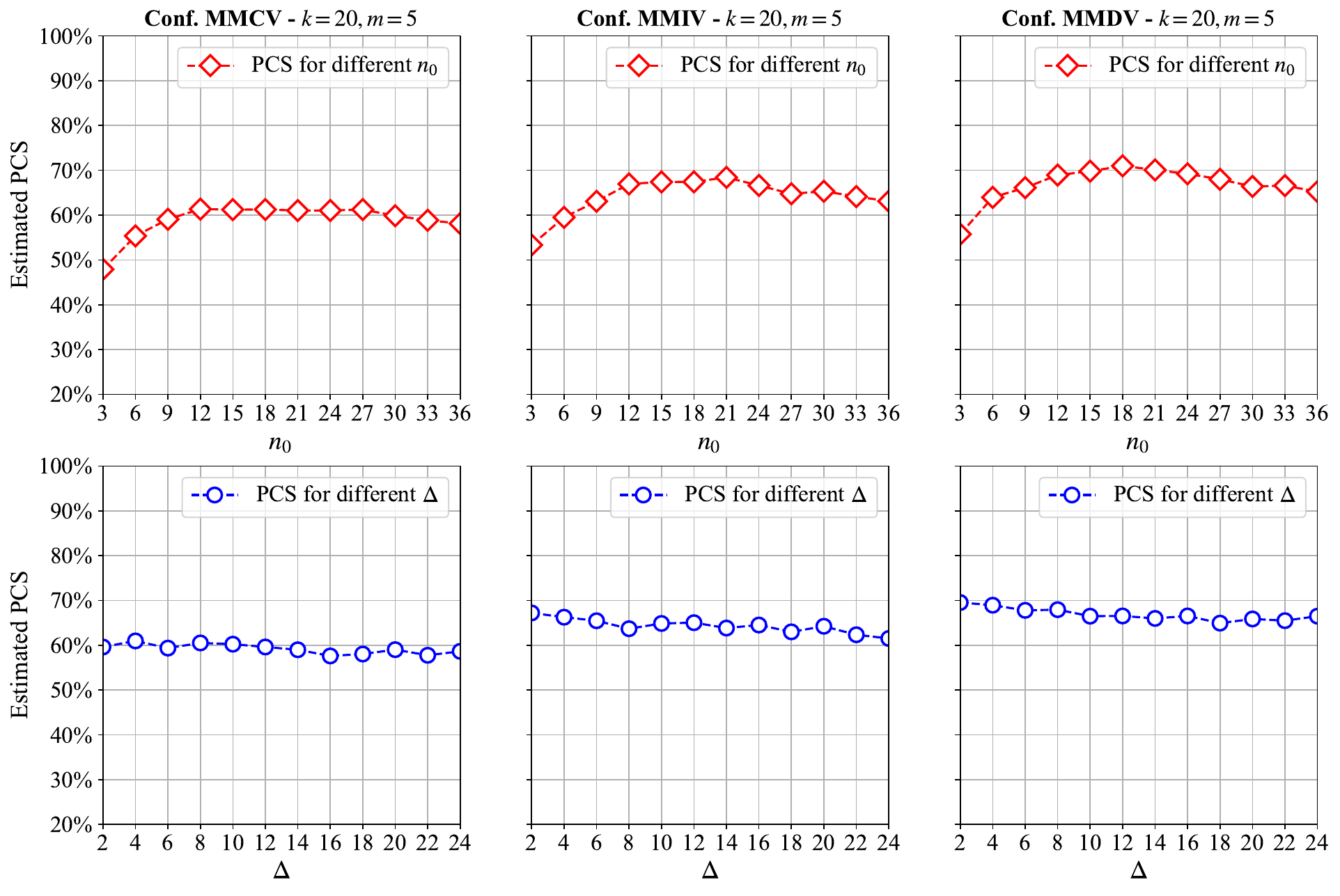}
    \caption{A comparison among the PCS of the AR-OCBA procedure under different $\mathbf{n_0}$ and $\mathbf{\Delta}$}
    \label{fig:robustness}
\end{figure}

{Figure \ref{fig:robustness} demonstrates the robustness of the AR-OCBA procedure's performance to the choices of \(n_0\) and \(\Delta\). For \(n_0\), the first row of Figure \ref{fig:robustness} shows that under all three configurations, the PCS of the AR-OCBA procedure exhibits an inverted U-shaped curve as \(n_0\) increases. Specifically, the PCS increases initially as \(n_0\) grows, reaches a near-optimal plateau over a wide range of \(n_0\) values, and then decreases for larger \(n_0\). This plateau suggests that in practice, it is not necessary to carefully tune \(n_0\). A reasonable choice between 10 and 15 may suffice to achieve near-optimal performance.}
For $\Delta$, the second row of Figure \ref{fig:robustness} shows that under all three configurations, the growth in $\Delta$ results in only mild changes in the PCS of the AR-OCBA procedure. These results imply that in practice, one does not need to carefully choose the values of $n_0$ and $\Delta$ when applying the AR-OCBA procedure.

\subsection{Sampling Budget Allocation}
\label{subsec: allocation}

In this subsection, we investigate the sampling budget allocation behavior of the AR-OCBA procedure. For clarity of presentation, we consider a small-scale problem with $k=3$ and $m=3$ where
$$
\left[\mu_{i j}\right]_{k \times m}=\left(\begin{array}{cccc}
            0.2 & 0.1  & 0.1 \\
            0.4 & 0.3  & 0.3 \\
            0.4 & 0.4 & 0.4
            \end{array}\right)
$$
and $\left[\sigma^2_{i j}\right]_{k \times m} = 1$. 
For the problem, the AR-OCBA procedure {appears to exhibit consistency} as the total sampling budget grows, as suggested by the trends of the PCS shown in Figure \ref{fig: small-consistency}. From the figure, we see that given a total sampling budget $N=5140km$, the PCS of the AR-OCBA procedure may be very close to 1. 
Then, we let $N=5140km$ and analyze the budget allocation behavior of the procedure in a single macro replication (sample path) that results in a correct selection. We plot the sample size of each scenario against the number of executed rounds of the AR-OCBA procedure and also the total allocated sample sizes of all alternatives after the total sampling budget is exhausted in Figure \ref{fig: budget_allocation}. 
\begin{figure}[htbp]
    \centering
    \includegraphics[width=0.45\linewidth]{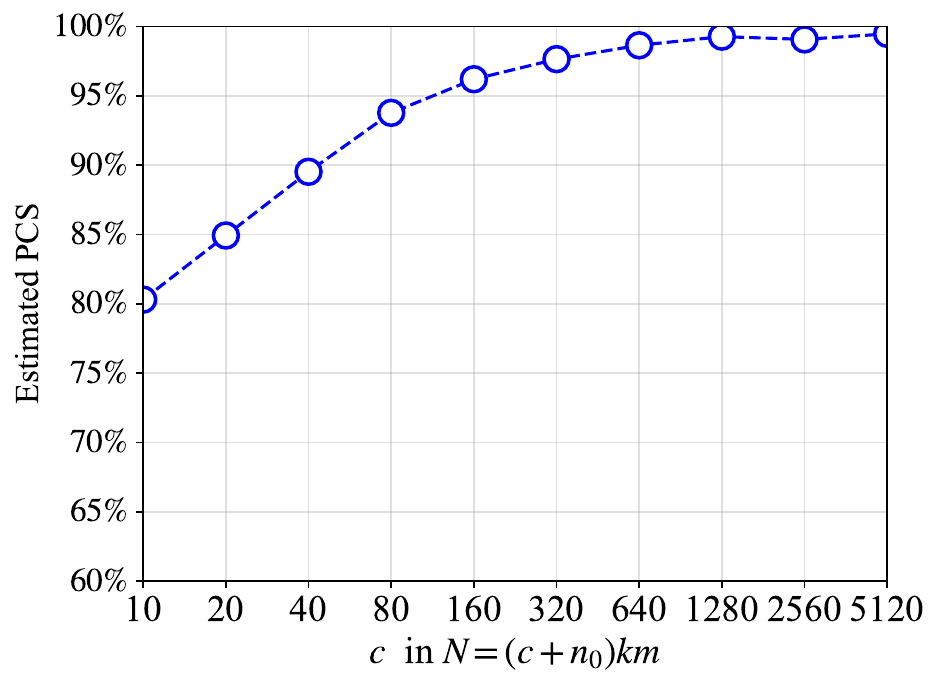}
    \caption{{An illustration of the PCS of the AR-OCBA procedure as c increases}}
    \label{fig: small-consistency}
\end{figure}
\begin{figure}[htbp]
    \centering
    \includegraphics[width=0.88\linewidth]{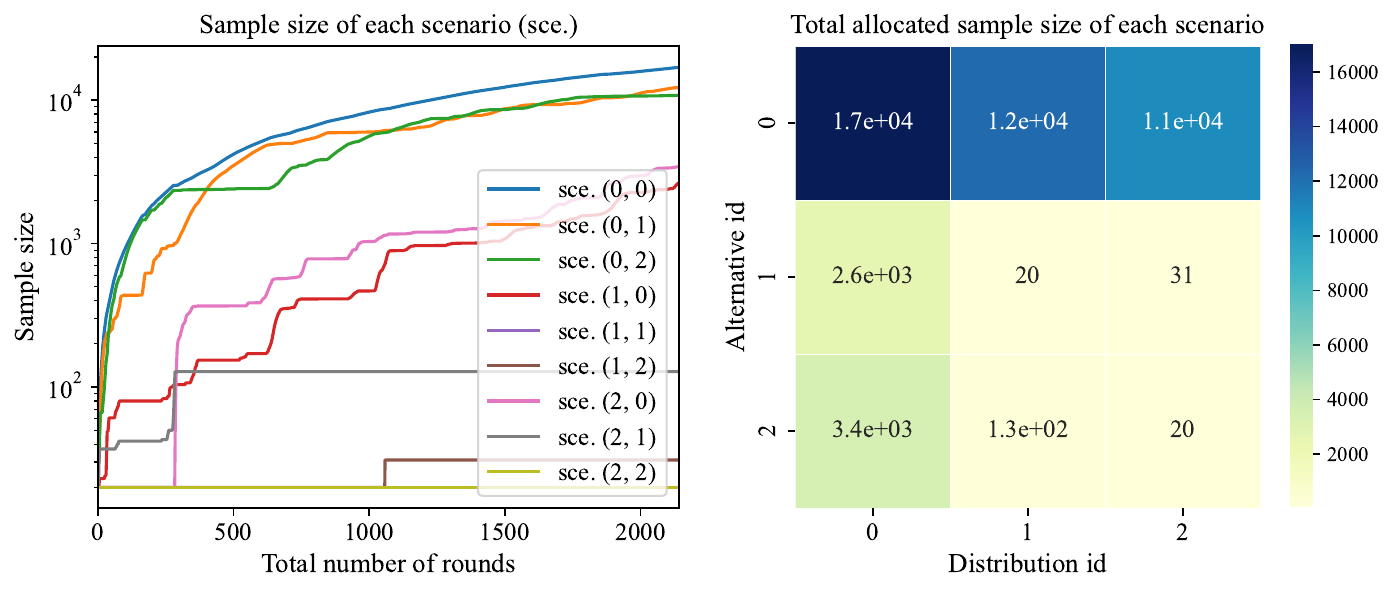}
    \caption{An illustration of the AR-OCBA procedure's budget allocation in a sample path}
    \label{fig: budget_allocation}
\end{figure}

From Figure \ref{fig: budget_allocation}, we may obtain an intuitive understanding of the sample allocation behavior of AR-OCBA. In the observed sample path, the AR-OCBA procedure concentrates only on $k+m-1=3+3-1=5$ scenarios when allocating the sampling budget. As shown in the sample size trajectories, only the scenarios $(0,0), (0,1), (0,2), (1, 0)$ and $(2,0)$ will constantly obtain new observations as the procedure proceeds. Consequently, when the sampling budget is exhausted, the sample sizes of these 5 scenarios far exceed the sample sizes of all other scenarios. Such sample allocation behavior is exactly as predicted by Theorem \ref{thm: OCBA}. 

Figure \ref{fig: budget_allocation} also reveals an important difference between the budget allocation behaviors of AR-OCBA procedure and the traditional OCBA procedure. It is known that in solving R\&S problems, the OCBA procedure will allocate a fixed non-zero proportion of the total sampling budget when it is sufficiently large, or formally speaking, approaches to infinity  \citep{wu2018analyzing, li2023convergence}. This means that as the sampling budget grows, all the alternatives will have the chance to obtain new observations. However, as discussed above, the AR-OCBA procedure may concentrate only on $k+m-1$ scenarios out of $km$ scenarios. This unique feature of selective sampling comes from the additive nature of the AR-OCBA procedure, and the demonstrated strong performance of the AR-OCBA procedure in this section highlights the importance of being additive in solving robust R\&S problems.

{This selective sampling behavior raises an question about the AR-OCBA procedure: does it truly achieve the asymptotic consistency, as suggested by Figure \ref{fig: small-consistency}? While the strong empirical performance of AR-OCBA demonstrates its effectiveness in finite-budget settings, its selective nature might pose challenges for achieving consistency in an asymptotic sense. In the R\&S literature, it is commonly understood that to achieve consistency, each alternative must receive an infinite number of observations as the total sampling budget grows to infinity, almost surely \citep{hong2021review}. However, as indicated by the budget allocation patterns shown in Figure \ref{fig: budget_allocation}, the AR-OCBA procedure may concentrate its sampling effort on only \(k + m - 1\) scenarios, potentially leaving some scenarios sampled only finitely often. This observation suggests that AR-OCBA may not satisfy the theoretical requirement of sampling all “alternatives” infinitely often, raising doubts about its asymptotic consistency. Investigating the asymptotic behavior of the AR-OCBA procedure theoretically would be an intriguing direction for future research. Addressing this potential issue could involve adopting exploration-enforcing strategies such as \(\epsilon\)-greedy exploration (see, e.g., \citealt{li2023convergence}). Specifically, the \(\epsilon\)-greedy strategy involves uniformly selecting a scenario to sample with a small probability \(0 < \epsilon < 1\) at each stage of the budget allocation process. By keeping \(\epsilon\) sufficiently small, this approach should not significantly affect the finite-budget performance of the AR-OCBA procedure while ensuring consistency.}



\vspace{0.1cm}
\section{Concluding Remarks}
\label{sec: conclusion}

In this paper, we provide a refined approach of the budget allocation for fixed-budget robust ranking and selection (R\&S) problems. By adopting an additive PICS bound of robust R\&S procedures, we derive a new optimal budget allocation scheme that minimizes this bound. Our derivation reveals a significant simplification in the allocation problem: the budget allocation for robust R\&S involving \( km \) scenarios can be mapped to the budget allocation for a traditional R\&S problem with \( k + m - 1 \) alternatives. This insight allows us to leverage established OCBA solutions for robust R\&S challenges. Accordingly, we propose the AR-OCBA procedure with a proportional stage-wise budget allocation rule. Numerical experiments validate the superiority of the AR-OCBA procedure, particularly as the problem scale increases. The observed performance improvements are consistent with our insights, confirming that the AR-OCBA procedure effectively allocates the limited sampling budget according to the additive structure of the new budget allocation scheme.

To end this paper, we point out several important directions worth of further investigation. 
{First, the numerical results in Section \ref{sec: numerical} suggest that the AR-OCBA procedure appears to exhibit  consistency. However, doubts remain regarding this property. Analyzing the asymptotic behavior of the procedure would be valuable for understanding its theoretical properties.}
Second, as our numerical analysis and \cite{gao2017robust} suggest, an additive structure of robust R\&S may be the key to efficiently solve robust R\&S problems. However, despite this insightful finding, there is currently no a theoretical characterization of this structure.
 Addressing this challenge would be valuable for sharpening the understanding towards robust R\&S. Lastly, the robust R\&S formulation adopted in this paper assumes a finite set of input distributions in the ambiguity set. However, it is possible that the ambiguity set includes infinite or even uncountable input distributions, e.g., ambiguity sets indexed by the values of continuous parameters of the input distribution. How to design efficient robust OCBA procedures to handle such ambiguity sets remains an open challenge.

\bibliographystyle{informs2014}
\bibliography{ref.bib}

\begin{thebibliography}{21}
\providecommand{\natexlab}[1]{#1}
\providecommand{\url}[1]{\texttt{#1}}
\providecommand{\urlprefix}{URL }

\bibitem[{Almomani \protect\BIBand{} Alrefaei(2016)}]{almomani2016ordinal}
Almomani MH, Alrefaei MH (2016) Ordinal optimization with computing budget allocation for selecting an optimal subset. \emph{Asia-Pacific Journal of Operational Research} 33(02):1650009.

\bibitem[{Ben-Tal \protect\BIBand{} Nemirovski(2002)}]{ben2002robust}
Ben-Tal A, Nemirovski A (2002) Robust optimization--methodology and applications. \emph{Mathematical programming} 92:453--480.

\bibitem[{Chen \protect\BIBand{} Lee(2011)}]{chen2011stochastic}
Chen CH, Lee LH (2011) \emph{Stochastic simulation optimization: an optimal computing budget allocation}, volume~1 (World scientific).

\bibitem[{Chen et~al.(2000)Chen, Lin, Y{\"u}cesan, \protect\BIBand{} Chick}]{chen2000simulation}
Chen CH, Lin J, Y{\"u}cesan E, Chick SE (2000) Simulation budget allocation for further enhancing the efficiency of ordinal optimization. \emph{Discrete Event Dynamic Systems} 10:251--270.

\bibitem[{Fan et~al.(2013)Fan, Hong, \protect\BIBand{} Zhang}]{fan2013robust}
Fan W, Hong LJ, Zhang X (2013) Robust selection of the best. \emph{2013 Winter Simulations Conference (WSC)}, 868--876 (IEEE).

\bibitem[{Fan et~al.(2020)Fan, Hong, \protect\BIBand{} Zhang}]{fan2020distributionally}
Fan W, Hong LJ, Zhang X (2020) Distributionally robust selection of the best. \emph{Management Science} 66(1):190--208.

\bibitem[{Gao et~al.(2016)Gao, Xiao, Zhou, \protect\BIBand{} Chen}]{gao2016optimal}
Gao S, Xiao H, Zhou E, Chen W (2016) Optimal computing budget allocation with input uncertainty. \emph{2016 Winter Simulation Conference (WSC)}, 839--846 (IEEE).

\bibitem[{Gao et~al.(2017)Gao, Xiao, Zhou, \protect\BIBand{} Chen}]{gao2017robust}
Gao S, Xiao H, Zhou E, Chen W (2017) Robust ranking and selection with optimal computing budget allocation. \emph{Automatica} 81:30--36.

\bibitem[{Hong et~al.(2021)Hong, Fan, \protect\BIBand{} Luo}]{hong2021review}
Hong LJ, Fan W, Luo J (2021) Review on ranking and selection: A new perspective. \emph{Frontiers of Engineering Management} 8(3):321--343.

\bibitem[{Kim \protect\BIBand{} Nelson(2006)}]{kim2006selecting}
Kim SH, Nelson BL (2006) Selecting the best system. \emph{Handbooks in operations research and management science} 13:501--534.

\bibitem[{Li \protect\BIBand{} Gao(2023)}]{li2023convergence}
Li Y, Gao S (2023) Convergence rate analysis for optimal computing budget allocation algorithms. \emph{Automatica} 153:111042.

\bibitem[{Liu et~al.(2023)Liu, Peng, Zhang, \protect\BIBand{} Zhou}]{liu2023efficient}
Liu X, Peng Y, Zhang G, Zhou R (2023) An efficient node selection policy for monte carlo tree search with neural networks. \emph{Available at SSRN 4450999} .

\bibitem[{Song \protect\BIBand{} Nelson(2017)}]{song2017input}
Song E, Nelson BL (2017) Input model risk. \emph{Advances in Modeling and Simulation: Seminal Research from 50 Years of Winter Simulation Conferences}, 63--80 (Springer).

\bibitem[{Song \protect\BIBand{} Nelson(2019)}]{song2019input}
Song E, Nelson BL (2019) Input--output uncertainty comparisons for discrete optimization via simulation. \emph{Operations Research} 67(2):562--576.

\bibitem[{Song et~al.(2015)Song, Nelson, \protect\BIBand{} Hong}]{song2015input}
Song E, Nelson BL, Hong LJ (2015) Input uncertainty and indifference-zone ranking \& selection. \emph{2015 Winter Simulation Conference (WSC)}, 414--424 (IEEE).

\bibitem[{Wan et~al.(2023)Wan, Hong, \protect\BIBand{} Fan}]{wan2023upper}
Wan Y, Hong LJ, Fan W (2023) Upper-confidence-bound procedure for robust selection of the best. \emph{2023 Winter Simulation Conference (WSC)}, 3647--3656 (IEEE).

\bibitem[{Wu \protect\BIBand{} Zhou(2018)}]{wu2018analyzing}
Wu D, Zhou E (2018) Analyzing and provably improving fixed budget ranking and selection algorithms. \emph{arXiv preprint arXiv:1811.12183} .

\bibitem[{Xu et~al.(2015)Xu, Huang, Chen, \protect\BIBand{} Lee}]{xu2015simulation}
Xu J, Huang E, Chen CH, Lee LH (2015) Simulation optimization: A review and exploration in the new era of cloud computing and big data. \emph{Asia-Pacific Journal of Operational Research} 32(03):1550019.

\bibitem[{Zhang \protect\BIBand{} Ding(2016)}]{zhang2016sequential}
Zhang X, Ding L (2016) Sequential sampling for bayesian robust ranking and selection. \emph{2016 Winter Simulation Conference (WSC)}, 758--769 (IEEE).

\bibitem[{Zhang \protect\BIBand{} Peng(2024)}]{zhang2024sample}
Zhang Z, Peng Y (2024) Sample-efficient clustering and conquer procedures for parallel large-scale ranking and selection. \emph{arXiv preprint arXiv:2402.02196} .

\bibitem[{Zhou \protect\BIBand{} Wu(2017)}]{zhou2017simulation}
Zhou E, Wu D (2017) Simulation optimization under input model uncertainty. \emph{Advances in Modeling and Simulation: Seminal Research from 50 Years of Winter Simulation Conferences}, 219--247 (Springer).

\end{thebibliography}

\end{document}